\documentclass[11pt,a4paper]{article}
\usepackage{authblk}

\usepackage{latexsym, amsmath, amsthm, amssymb}
\usepackage[usenames,dvipsnames]{color}
\usepackage{graphicx}
\usepackage{microtype}
\usepackage{cite}

\newtheorem{theorem}{Theorem}
\newtheorem{lemma}[theorem]{Lemma}
\newtheorem{corollary}[theorem]{Corollary}

\theoremstyle{definition}
\newtheorem{definition}{Definition}

\newtheorem*{rem}{Remark}

\title{On the Relation of KS Entropy and Permutation Entropy}
\author[1]{Karsten Keller\thanks{Corresponding address: Institute of Mathematics, University of L\"ubeck,
Ratzeburger Alley 160, Building 64, 23562 L\"ubeck, Germany. Tel.: +49 451 500 3165; fax: +49 451 500 3373; e-mail: keller@math.uni-luebeck.de (K. Keller)}}
\author[1,2]{Anton M.~Unakafov}
\author[1,2]{Valentina A.~Unakafova}
\affil[1]{Institute of Mathematics, University of L\"ubeck}
\affil[2]{Graduate School for Computing in Medicine and Life Sciences, University of L\"ubeck}
\date{May 19, 2012}
\begin{document}
\maketitle

\begin{abstract}
\noindent Since Bandt et al.~have shown that the permutation entropy and the Kolmogo\-rov-Sinai entropy coincide for piecewise monotone interval maps,
the relationship of both entropies for time-discrete dynamical systems is of a certain interest.
The aim of this paper is a discussion of this relationship on the basis of an ordinal characterization of the Kolmogorov-Sinai entropy recently given.
\\

\end{abstract}
\section{Introduction}\label{intro}

\subsection{State of the art}
In their seminal paper \cite{BandtKellerPompe2002} Bandt et al.~have given a characterization of the
Kolmogorov-Sinai entropy (KS entropy) of a piecewise monotone interval map
on the basis of quantifying ordinal patterns in the dynamics of the map.
The central concept in their work is the permutation entropy introduced in \cite{BandtPompe2002}.

This concept, which was also generalized to the multidimensional case (see \cite{KellerSinn2010, Keller2011}),
allows a relatively simple and robust quantification of the complexity of a dynamical system.
In the case of piecewise monotone interval maps the permutation entropy coincides with the KS entropy \cite{BandtKellerPompe2002}.
Note that the consi\-deration of ordinal pattern distributions underlying a dynamical system provides interesting insights into the structure of the system.
For a general discussion, see Amigo \cite{Amigo2010}.

The relationship of KS entropy and permutation entropy is the central point of interest of this paper.
Whereas KS entropy has been shown to be not larger than permutation entropy
(see Keller and Sinn \cite{KellerSinn2009, KellerSinn2010, Keller2011} and Amigo et al.~\cite{AmigoKennelKocarev2005,Amigo2012}),
to our knowledge there is nothing known about the equality of these entropies beyond the case of piecewise monotone interval maps.
(Note that Amigo et al.~\cite{AmigoKennelKocarev2005,Amigo2012} have shown equality of KS entropy and permutation entropy for a concept of permutation entropy
that is qualitatively different from the one originally given.)
Here we discuss the relationship of KS entropy and permutation entropy from a structural viewpoint using an ordinal characterization of KS entropy recently provided in \cite{KellerSinn2009, KellerSinn2010, Keller2011}.

\subsection{Preliminaries}
%\subsection{The entropy and entropy rate of a finite partition}
In the whole paper $\left( \Omega, \mathbb{B}(\Omega), \mu, T \right)$ is a measure-preserving dynamical system, where
$\Omega$ is a non-empty topological space,
$\mathbb{B}(\Omega)$ is the Borel sigma-algebra on it,
$\mu:\mathbb{B}(\Omega) \rightarrow [0,1]$ is a probability measure
and $T: \Omega \hookleftarrow$ a $\mathbb{B}(\Omega)$-$\mathbb{B}(\Omega)$-measurable $\mu$-preserving map,
i.e.  $\mu(T^{-1}(B)) = \mu(B)$  for all $B \in \mathbb{B}(\Omega)$.

The {\it (Shannon) entropy} of a finite partition ${\mathcal P} = \{P_1, P_2, \ldots ,P_l\} \subset \mathbb{B}(\Omega)$ of $\Omega$ is defined by
\begin {equation*}
    H({\mathcal P}) = - \sum_{P \in {\mathcal P}} \mu(P) \ln \mu(P)
\end {equation*}
(with $\ln 0:=0$).

%%\subsection{The entropy rate of a finite partition}
Given a finite partition ${\mathcal P} = \{P_1, P_2, \ldots ,P_l\} \subset \mathbb{B}(\Omega)$ of $\Omega$, consider the corresponding alphabet $A = \lbrace 1, 2, \ldots,l\rbrace$.
%By assigning to each point $\omega\in\Omega$ the symbol $a\in A$ with $\omega\in P_a$ and in the same manner to each of its iterates the corresponding symbol,
By assigning to each point $\omega\in P_a$ the symbol $a\in A$ and in the same manner to each of its iterates the corresponding symbol,
the dynamical system can be described by the language consisting of words over $A$ obtained from successive iterates.
Roughly speaking, the more complex this language is, the more complex is the dynamical system.
It is however necessary to consider different partitions in order to measure the `truth'.

Classifying points according to the words $a_1 a_2\ldots a_n$ obtained from the initial parts of their orbits, for each $n\in {\mathbb N}$ one obtains a partition
${\mathcal P}_n$ consisting of the sets
\begin{equation}\label{partdef}
    P_{a_1 a_2 ... a_n} = \lbrace \omega\in\Omega\,\mid\,\omega \in P_{a_1}, T(\omega) \in P_{a_2}, ..., T^{\circ n-1}(\omega) \in P_{a_n} \rbrace
\end{equation}
for $a_1,a_2,\ldots ,a_n\in A$. Here $T^{\circ t}$ denotes the $t$-th iterate of $T$.
The original partition coincides with ${\mathcal P}_1$, and the larger $n$ is the finer is the partition ${\mathcal P}_n$.

The {\it entropy rate} of $T$ with respect to $\mu$ and the partition ${\mathcal P}$ is given by
\begin {equation*}
    h_\mu(T, {\mathcal P}) = \lim_{n \rightarrow \infty} \frac{H({\mathcal P}_n)}{n}
    \text{.}
\end {equation*}
This limit is well-defined (see, e.g.~\cite{Walters82}).

\subsection{Kolmogorov-Sinai entropy and permutation entropy}
The {\it Kolmogo\-rov-Sinai entropy} ({\it KS entropy}) of $T$ with respect to $\mu$ is defined by
\begin {equation*}
    h_\mu(T) = \sup_{{\mathcal P} \subset \mathbb{B}(\Omega)\text{ finite partition of } \Omega}  h_\mu(T, {\mathcal P})
    \text{.}
\end {equation*}

Roughly speaking, it is the `maximal' possible information of the dynamical system that can be obtained from a symbolization by a finite alphabet.

It is often not easy to determine KS entropy,
since in the general case it is impossible to check all finite partitions of $\Omega$.
In a small number of cases one can find a generating partition ${\mathcal G}$, for which by the Kolmogorov-Sinai theorem it holds $h_\mu(T) = h_\mu(T, {\mathcal G})$.

Given a random vector ${\bf X}=(X_1,X_2,\ldots,X_N)$ on $(\Omega,{\mathbb B}(\Omega))$ with~$X_1,X_2,\linebreak\ldots,X_N: \Omega \rightarrow R$,
the {\it permutation entropy} $h_\mu^{\bf X}(T)$ with respect to ${\bf X}$ is defined by
\begin{align}\label{permEntropy}
    h_\mu^{\bf X}(T)=\varlimsup_{d\to\infty} \frac{H({\mathcal P}^{\bf X}(d))}{d},
\end{align}
where $({\mathcal P}^{\bf X}(d))_{d\in {\mathbb N}}$ is an increasing sequence of special finite partitions ${\mathcal P}^{\bf X}(d)\subset {\mathbb B}(\Omega)$ of $\Omega$
determined by the collection ${\bf X}$ of `observables' on the basis of considering order relations.
(Increasing means that ${\mathcal P}^{\bf X}(d\hspace{0.35mm}')$ is a refinement of ${\mathcal P}^{\bf X}(d)$ for $d\hspace{0.2mm}'\geq d$,
e.a.~each set in ${\mathcal P}^{\bf X}(d\hspace{0.2mm}')$ is contained in a set in ${\mathcal P}^{\bf X}(d)$.)

We do not provide the detailed description of the partitions ${\mathcal P}^{\bf X}(d)$ at this point.
Instead we refer to Definition \ref{opdef}.
The most important fact is that for certain choices of ${\bf X}$ these partitions determine the KS entropy of $T$
(compare \cite{KellerSinn2009,KellerSinn2010,Keller2011}):
\begin{align} \label{linkToKS}
    h_\mu (T)=\lim_{d\to\infty} h_\mu (T,{\mathcal P}^{\bf X}(d)).
\end{align}
For the possible choices of ${\bf X}$, see Theorems \ref{choice1}, \ref{choice2} and \ref{choice3}.

Having a closer look at the structure of \eqref{permEntropy} and \eqref{linkToKS}, we want to consider the following general problem:
\begin{align*}
    \text{When }h_\mu (T)=h_\mu^{\bf X}(T)\text{?}
\end{align*}

This paper is organized as follows. 
In Section~\ref{framework} we discuss the above problem in an abstract framework and establish Theorem \ref{the01} with Corollary~\ref{cor03} being
the main result of this paper.
Section~\ref{ordinal} gives the detailed descriptions of the partitions mentioned above and of the ordinal patterns on the basis of which these partitions are defined. 
Moreover, we provide conditions under which \eqref{linkToKS} is valid.
In Section~\ref{proof} we prove Theorem \ref{the01}.

\section{General framework}\label{framework}
It is useful to put the discussion into a more abstract context.
Let $({\mathcal P}(d))_{d\in {\mathbb N}}$ be a sequence of partitions ${\mathcal P}(d)\subset \mathbb{B}(\Omega)$ of $\Omega$ for which
\begin{align}\label{exist}
    \lim_{d\to\infty} h_\mu (T,{\mathcal P}(d))\text{ exists.}
\end{align}
Further, assume that
\begin{align}\label{finer}
    {\mathcal P}(d+n-1)\text{ is finer than }{\mathcal P}(d)_n\text{ for all }d,n\in {\mathbb N}\text{ with }n>2.
\end{align}

For the following, we only need \eqref{exist} and \eqref{finer} and have mainly a partition
$({\mathcal P}(d))_{d\in {\mathbb N}}=({\mathcal P}^{\bf X}(d))_{d\in {\mathbb N}}$ for some random vector ${\bf X}$ on $(\Omega,{\mathbb B}(\Omega))$ in mind.
(Given such a partition, \eqref{exist} is satisfied since $({\mathcal P}^{\bf X}(d))_{d\in {\mathbb N}}$ is increasing,
and \eqref{finer} holds according to Lemma \ref{finerlem} given in Section \ref{ordinal}.)

Under \eqref{exist} and \eqref{finer}, we interpret an element of the partition $({\mathcal P}(d))_{d\in {\mathbb N}}$
as the set of all $\omega\in\Omega$ providing a certain dynamical pattern of some length $d$.
(The starting point is not counted.)

Let $H_n(d):=H({\mathcal P}(d)_n)$.
Then $\frac{H_1(d)}{d}$ can be interpreted as the mean information per iterate contained in a pattern of length $d$ and
$\frac{H_2(d)}{1+d}$ that was contained in two successive patterns of length $d$
(taking into account that they describe $d$ identical iterates and that the second pattern holds new information about only one iterate).
More generally, $\frac{H_n(d)}{n+d-1}$ can be interpreted as the mean information per iterate contained in $n$ successive patterns.
Then
\begin{align}
    h(d):=\lim_{n\to\infty}\frac{H_n(d)}{n+d-1}=\lim_{n\to\infty}\frac{H_n(d)}{n}\label{hd}
\end{align}
is the entropy rate of $T$ with respect to ${\mathcal P}(d)$. Furthermore, by \eqref{finer} it holds
\begin{align}\label{finer2}
    H_n(d) \leq H_1(d + n - 1) \text{ for all }n, d \in \mathbb{N}.
\end{align}

We want to consider quantities $h$ and $h^\ast$ defined by
\begin{align}
    h:=\lim_{d\to\infty}h(d)=\lim_{d\to\infty}\lim_{n\to\infty}\frac{H_n(d)}{n},\label{h}\\
    h^\ast:=\varlimsup_{d\to\infty} \frac{H_1(d)}{d}.\label{hast}
\end{align}
In Section \ref{proof} we will prove the following statement:
\begin{theorem}\label{the01}
	For $n,d\in {\mathbb N}$, let $H_n(d), h(d), h, h^\ast$ be non-negative real numbers satisfying \eqref{hd}--\eqref{hast}.
    (`\,$:=$' has to be considered as `\,$=$'.)
	%\eqref{hd}, \eqref{finer2}, \eqref{h} and \eqref{hast}.
	Then it holds $h\leq h^\ast$ and the following statements are equivalent:
	\begin{enumerate}
		\item [(i)]  $h = h^\ast$.
		\item [(ii)] For each $\varepsilon > 0$ there exist some $d_\varepsilon \in \mathbb{N}$ and some $M_\varepsilon \in \mathbb{R}$, such
			     that for each $d \geq d_\varepsilon$ the following holds:
			      \begin {align}
				    \frac{H_n(d)}{d + n - 1} < M_\varepsilon\text{ for all }n \in \mathbb{N}.
				    \label{eq:T1_statement1}\\
				    \text{There is some }n_d \in \mathbb{N}\text{ with }\frac{H_n(d)}{d + n - 1} > M_\varepsilon - \varepsilon\text{  for all $n \geq n_d$}
				    \text{.}
				    \label{eq:T1_statement2}
			      \end {align}
		\item [(iii)] For each $\varepsilon > 0$ there exists some $d_\varepsilon \in \mathbb{N}$ such that
			     for all $d \geq d_\varepsilon$ there is some $n_d \in \mathbb{N}$ with
			     \begin {equation}
				  H_1(d + n - 1) - H_n(d) < (n - 1)\varepsilon \text{ for all }n \geq n_d.
				  \label{eq:T1_statement3}
			     \end {equation}	
	\end{enumerate}
\end{theorem}

As an immediate consequence of Theorem \ref{the01} one gets
\begin{corollary} \label{cor02}
	Let $({\mathcal P}(d))_{d\in {\mathbb N}}$ be a sequence of finite partitions ${\mathcal P}(d)\subset {\mathbb B}(\Omega)$ of $\Omega$ with \eqref{exist} and \eqref{finer}.
	Then it holds
	\begin {equation*}
	    \lim_{d\to\infty} h_\mu (T,{\mathcal P}(d))\leq\varlimsup_{d\to\infty} \frac{H({\mathcal P}(d))}{d}.
	\end {equation*}
	Moreover, the equality
	\begin {equation} \label{meq}
	    \lim_{d\to\infty} h_\mu (T,{\mathcal P}(d)) = \varlimsup_{d\to\infty} \frac{H({\mathcal P}(d))}{d}
	\end {equation}
	holds iff for each $\varepsilon > 0$ there exists some $d_\varepsilon \in \mathbb{N}$ such that for all $d \geq d_\varepsilon$ there is some $n_d \in \mathbb{N}$ with
	\begin {equation*}
	    H({\mathcal P}(d+n-1)) - H({\mathcal P}(d)_n) < (n-1)\varepsilon \text{ for all }n \geq n_d.
	\end {equation*}	
\end{corollary}

\begin{rem}
    Inequality \eqref{eq:T1_statement3} can be rewritten as follows:
    \begin {equation*}
	    \sum_{k=1}^{n-1} \left( H_{k}(d+n-k) - H_{k+1}(d+n-(k+1)) \right) < (n - 1)\varepsilon \text{ for all }n \geq n_d.
    \end {equation*}	
    From this representation one can see that under the assumptions of Corollary \ref{cor02}
    the following statement presents a potentially helpful sufficient condition for \eqref{meq}. 
    For each $\varepsilon > 0$ there exists some $d_\varepsilon \in \mathbb{N}$ such that for all $d \geq d_\varepsilon$ it holds
    \begin {equation} \label{suffCondition}
	 H({\mathcal P}(d + 1)_{n-1}) - H({\mathcal P}(d)_n) < \varepsilon \text{ for all }n>1.
    \end {equation}
    Note that if ${\mathcal P}(d + 1)_{n-1}$ is finer than ${\mathcal P}(d)_n$ (for the partitions ${\mathcal P}^{\bf X}(d); d\in {\mathbb N}$ this holds according to Lemma \ref{finerlem}),
    the left-hand part of \eqref{suffCondition} is no more than the conditional entropy $H({\mathcal P}(d+1)_{n-1} \mid {\mathcal P}(d)_{n})$,
    i.e. the amount of new information obtained from ${\mathcal P}(d+1)_{n-1}$ given that obtained from ${\mathcal P}(d)_{n}$.
\end{rem}

For the relationship of KS entropy and permutation entropy, Corollary \ref{cor02} provides the following
\begin{corollary} \label{cor03}
	For each ${\mathbb R}$-valued random vector ${\bf X}=(X_1,X_2,\ldots ,X_N)$ on $(\Omega,{\mathbb B}(\Omega))$, it holds
	\begin {equation*}
	    \lim_{d\to\infty} h_\mu (T,{\mathcal P}^{\bf X}(d))\leq h_\mu^{\bf X}(T).
	\end {equation*}
	Moreover, if \eqref{linkToKS} is valid, then the following statements are equivalent:
	\begin{enumerate}
		\item[(i)] $h_\mu (T) = h_\mu^{\bf X}(T)$.
		\item[(ii)] For each $\varepsilon > 0$ there exists some $d_\varepsilon \in \mathbb{N}$ such that
			    for all $d \geq d_\varepsilon$ there is some $n_d \in \mathbb{N}$ with
			    \begin {equation*}
				H({\mathcal P}^{\bf X}(d+n-1)) - H({\mathcal P}^{\bf X}(d)_n) < (n-1)\varepsilon \text{ for all }n \geq n_d.
			    \end {equation*}
	\end{enumerate}
\end{corollary}
Corollary \ref{cor03} has to be considered together with Theorem \ref{choice3} describing cases where \eqref{linkToKS} is satisfied.

In order to illustrate Corollary \ref{cor03}, we present results for the logistic map $T:[0,1]\hookleftarrow$ defined by $T(x)=4x(1-x)$, based on numerical simulation.
Here $\mu$ is assumed to be the invariant measure with density $\frac{1}{\pi\sqrt{x(1-x)}}$ with respect to the equidistribution on $[0,1]$ and $X: [0,1]\rightarrow {\mathbb R}$ to be the identity.

Estimations of $H_n(d)$ for different values of $d$ and $n$ are given by the corresponding empirical entropies, computed from an orbit of length $10^8$ of a pseudo-random point in $[0,1]$ (with respect to the equidistribution). Since the logistic map is ergodic with respect to $\mu$, distributions of such orbits differ only slightly.
We restricted $d+n-1$ to maximally $16$ because computer memory consumption increases fast with the growth of $(d+n)$.

According to the result of Bandt et al.~\cite{BandtKellerPompe2002}, statement (i) of Corollary \ref{cor03} is valid for $T$, hence statement (ii) is valid too.
Figure \ref{figure1} illustrates (slow) convergence of $\frac{H_1(d)}{d}$ to the KS entropy of $T$, which is equal to $\ln 2$ (see e.g.~\cite{Amigo2010}).
Figure \ref{figure2} illustrates that the terms $\frac{H_1(d+n-1)-H_n(d)}{n}$ converge (slowly) to zero for increasing $n$.
\begin{figure}[h]
      \centering
      \includegraphics[scale=0.6]{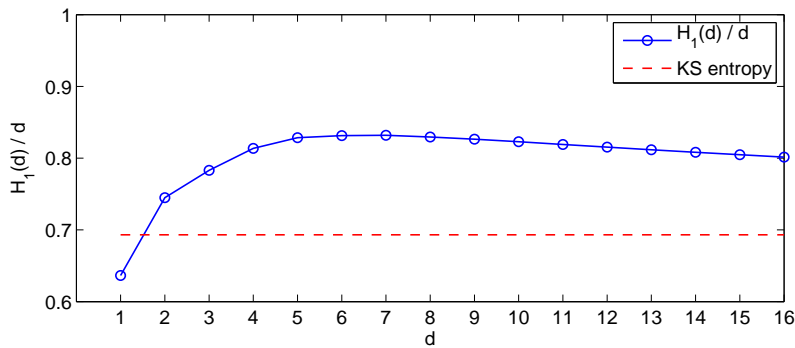}
      \caption{ Values of $\frac{H_1(d)}{d}$ in comparison to the KS entropy (dotted line)}
      \label{figure1}
\end{figure}
\begin{figure}[h]
      \centering
      \includegraphics[scale=0.6]{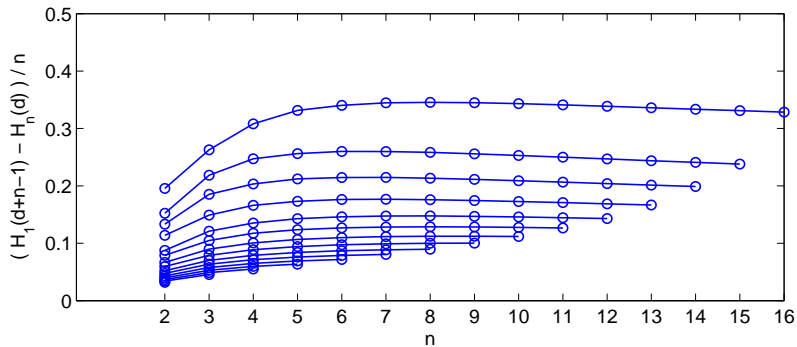}
      \caption{ Values of $\frac{H_1(d+n-1)-H_n(d)}{n}$ for logistic map for $d=1,2,\ldots, 16$ (from lower to upper curves) }
      \label{figure2}
\end{figure}

\section {Ordinal patterns, ordinal partition and permutation entropy}\label{ordinal}
The purpose of this section is to provide a brief review of ordinal patterns and on this basis to describe in detail the partitions considered in the Introduction.
Moreover, we summarize results relating KS and permutation entropies.

\subsection{Ordinal patterns}
We start from the definition of ordinal patterns.
Note that different authors determine this notion slightly differently.
Here we follow the definition given in \cite{Keller2011}.

\begin {definition}
    For $d \in \mathbb{N}$ denote the set of permutations of $\lbrace 0, 1, 2, ..., d\rbrace$ by $\Pi_d$.
    We say that a real vector $(x_0, x_1, ..., x_d)\in {\mathbb R}^{d+1}$
    has {\it ordinal pattern $\pi = (r_0, r_1,\ldots, r_d) \in \Pi_d$ of order $d$} if
    \begin {equation*}
	x_{r_0} \geq x_{r_1} \geq ... \geq x_{r_d}
    \end {equation*}
    and
    \begin {equation}\label{equalco}
	r_{l-1} > r_{l} \text{ in the case } x_{r_{l-1}} = x_{r_{l}}
	\text{.}
     \end {equation}
\end {definition}

\begin{rem}
Note that from the set of real numbers we only use that it is totally ordered; in other words, ordinal patterns can be defined for finite sequences with elements from any totally ordered set.
\end{rem}

Ordinal patterns describe all order relations between the components of a $(d+1)$-dimensional vector.
In this sense, permutations are only used as a representation of the `order type' of a vector, rather natural, but not unique (see, for instance, \cite{HarunaNakajima2011,KellerSinnEmonds2007}).
By definition, there are $(d+1)!$ different such patterns.
Also note that the treatment of equal values according to \eqref{equalco} is arbitrary, but convenient from the computational viewpoint (see \cite{KellerSinnEmonds2007}).
In many situations the probability of equal values is (near to) zero making treatment of equality redundant.

\subsection{Ordinal partitions}
Given a measure preserving dynamical system and a collection of `observables' on the system, we define now a partition for each $d\in {\mathbb N}$.

\begin{definition}\label{opdef}
  For $N\in {\mathbb N}$, let ${\bf X}=(X_1,X_2,\ldots ,X_N)$ be a ${\mathbb R}$-valued random vector on $(\Omega,{\mathbb B}(\Omega))$.
  Then, for $d\in {\mathbb N}$, the partition
  \begin{align*}
	{\mathcal P}^{\bf X}(d)=\{P_{(\pi_1,\pi_2,\ldots ,\pi_N)}\,\mid\,\pi_i\in\Pi_d\text{ for }i=1,2,\ldots ,N\}
  \end{align*}
  \hspace*{5mm} with
  \begin{align*}
      \vspace{5mm}P_{(\pi_1,\pi_2,\ldots ,\pi_N)}=\{\omega\in\Omega\,\mid\,(X_i(T^{\circ d}(\omega)),X_i(T^{\circ {d-1}}(\omega)),\ldots ,X_i(T(\omega)),X_i(\omega))\hspace*{2.3mm}\\
      \text{has ordinal pattern }\pi_i\text{ for }i=1,2,\ldots ,N\}
  \end{align*}
  is called {\it ordinal partition of order $d$} with respect to $T$ and ${\bf X}$.
\end{definition}

${\mathcal P}^{\bf X}(d)$ classifies the points of $\Omega$ according to the ordinal patterns `measured' by the `observables' $X_i$.
The vectors, from which the ordinal patterns are taken, are considered in inverse time order, ensuring compatibility with previous related papers.
Here the idea was to use only present values and values from the past in order to be `causal'.

Note that the partition ${\mathcal P}^{\bf X}(d)$ corresponds to the alphabet $(\Pi_d)^N$ and ${\mathcal P}^{\bf X}(d)_n$ to the set of words of length $n$ over the alphabet $(\Pi_d)^N$ (compare \eqref{partdef}).

\begin{lemma}\label{finerlem}
    Let {\bf X} be given as in Definition~\ref{opdef}, and for $n,d\in {\mathbb N}$ let ${\mathcal P}(d)_n={\mathcal P}^{\bf X}(d)_n$.
    Then, for all $n>1$, the partition ${\mathcal P}(d+1)_{n-1}$ is finer than the partition ${\mathcal P}(d)_n$.
    Moreover, it holds \eqref{finer}.
\end{lemma}
\begin{proof}
Let $d,n\in {\mathbb N}$ with $n>1$ and let $\omega,\omega'\in\Omega$ be in different sets of the partition ${\mathcal P}(d)_n$.
Then there exist some $i\in\{1,2,\ldots ,N\}$ and some $t\in \{0,1,\ldots ,n\}$ such that the ordinal patterns of order $d$ of
$( X_i(T^{\circ t+d}(\omega )), \linebreak X_i(T^{\circ {t+d-1}}(\omega )),\ldots, X_i(T^{\circ t}(\omega )) )$ and
$( X_i(T^{\circ t+d}(\omega')),X_i(T^{\circ {t+d-1}}(\omega')),\ldots, \linebreak  X_i(T^{\circ t}(\omega')) )$ are different.

From this one easily sees that for $k=n-1$ if $t=n$, and for $k = t$ otherwise, ordinal patterns of order $(d+1)$ of vectors
$( X_i(T^{\circ k+d+1}(\omega )), X_i(T^{\circ {k+d}}(\omega )),\ldots, \linebreak X_i(T^{\circ k}(\omega )) )$ and
$( X_i(T^{\circ k+d+1}(\omega')), X_i(T^{\circ {k+d}}(\omega')),\ldots,	      	   X_i(T^{\circ k}(\omega')) )$ are different.
This shows that $\omega,\omega'\in\Omega$ are in different sets of the partition ${\mathcal P}(d+1)_{n-1}$.
Therefore, ${\mathcal P}(d+1)_{n-1}$ is finer than ${\mathcal P}(d)_n$.

We also have that ${\mathcal P}(d+2)_{n-2}$ is finer than ${\mathcal P}(d+1)_{n-1}$,
${\mathcal P}(d+3)_{n-3}$ is finer than ${\mathcal P}(d+2)_{n-2}$, .\,.\,. ,
and ${\mathcal P}(d+n-1)={\mathcal P}(d+n-1)_1$ is finer than ${\mathcal P}(d+n-2)_2$.
This provides \eqref{finer}.
\end{proof}

\subsection{KS entropy on the basis of ordinal partitions}
Our discussion of the relation of KS entropy and permutation entropy was based on equality \eqref{linkToKS}.
In Theorems \ref{choice1}, \ref{choice2} and \ref{choice3} we summarize statements from \cite{KellerSinn2009, KellerSinn2010, Keller2011} guaranteeing this equality.
\begin{theorem}\label{choice1}
    For $N\in {\mathbb N}$, let ${\bf X}=(X_1,X_2,\ldots , X_N)$ be a random vector on $(\Omega,{\mathbb B}(\Omega))$.
    Then $h_\mu (T)=\lim\limits_{d\to\infty} h_\mu (T,{\mathcal P}^{\bf X}(d))$ is valid in each of the following two cases.
    \begin{enumerate}
	  \item[(i)] $\Omega$ is a Borel subset of ${\mathbb R}^N$ and $X_i$ is the $i$-th coordinate projection for $i=1,2,\ldots ,N$,
		     i.e.~$X_i((\omega_1,\omega_2,\ldots ,\omega_N))\!=\!\omega_i$ for $(\omega_1,\omega_2,\ldots,\omega_N)\in\Omega $.
	  \item[(ii)] $\Omega$ is a compact Hausdorff space, $X_i$ is continuous for $i=1,2,\ldots ,N$, and ${\bf X}$ is injective.
    \end{enumerate}
\end{theorem}

In special cases, according to Takens' embedding theory, the KS entropy of a multidimensional system can be obtained from only a one-dimensional measurement.
For the background, see Takens \cite{Takens81} and Sauer \cite{Sauer91}.

Let us recall that a property is said to be {\it generic} for a topological space if it holds for all points of an open dense subset of the space.
Moreover, a property is called {\it prevalent} for a topological vector space $V$ over ${\mathbb R}$ if it holds for all points of a Borel set $A$ with the following property:
There exists a finite-dimensional subspace $W$ of $V$ such that for all $v\in V$, the point
$v+w$ belongs to $A$ for Lebesgue-a.a.~$w\in W$.
Note that for a finite-dimensional $V$ of dimension $m$ a property is prevalent if it holds for Lebesgue-a.a.~points of $V$. 
This follows from from Fubini's theorem and shows that prevalence is a generalization of
Lebesgue-a.a.~to the infinite-dimensional case.

\begin{theorem}\label{choice2}
    Let $\Omega$ be a compact $C_2$-manifold of some dimension $m\in {\mathbb N}$.
    Then for the set of pairs $(S,{\bf X})$ of $C_2$-diffeomorphisms $S\!:\Omega\!\hookleftarrow$ and
    $C_2$-maps ${\bf X}:\Omega\rightarrow {\mathbb R}$ equipped with the $C_1$-topology,
    the following property is generic:

    If $\nu\!:\mathbb{B}(\Omega) \rightarrow [0,1]$ is an $S$-invariant probability measure,
    then $h_\nu(S)=\!\lim\limits_{d\to\infty}\! h_\nu(S,{\mathcal P}^{\bf X}(d))$.
\end{theorem}

\begin{theorem}\label{choice3}
    Assume that $k\in {\mathbb N}$, that $\Omega$ is a compact set contained in some open subset $U\subset {\mathbb R}^k$, that $\Omega$ has box
    dimension $m$, and that $T$ is the restriction of a $C_1$-diffeomorphism $\widetilde{T}$ on $U$ to $\Omega$.

    Further, assume the existence of some $N>2m$ in ${\mathbb N}$ such that for each $p\leq N$ in ${\mathbb N}$
    the set $\Omega_p$ of periodic points in $\Omega$ of period $p$ has the following properties:
    \begin{multline*}
	    \hspace{1cm}\text{The box dimension of }\Omega_p\text{ is less than }\frac{p}{2},\\
	    \shoveleft{\hspace{1cm} \text{the linearization of $\widetilde{T}^{\circ p}$ at each $\omega \in \Omega_p$ has distinct eigenvalues.} }
    \end{multline*}
    Then for the vector space of real-valued $C_1$-maps ${\bf X}$ equipped with the $C_1$-topology the property $h_\mu(T)=\lim\limits_{d\to\infty} h_\mu(T,{\mathcal P}^{\bf X}(d))$ is prevalent.
\end{theorem}

\subsection{Another concept of permutation entropy}
As mentioned in the Introduction, Amigo et al.~\cite{AmigoKennelKocarev2005,Amigo2012} have introduced a version of permutation entropy being
qualitatively different from the originally given one, but being interesting in its own right:
They first show that the KS entropy of a stochastic finite symbol source can be considered as a permutation entropy when the symbols are totally ordered arbitrarily
(see \cite{AmigoKennelKocarev2005}, and, for an alternative proof of the corresponding statement, see Haruna and Nakajima \cite{HarunaNakajima2011}).
Then they take the limit of the permutation entropies for the symbolizations obtained from finer and finer finite partitions of the state space $\Omega$.

We do not want to go into the detail, but we give a reformulation of the main result in \cite{Amigo2012} that is based on one-dimensional `observables':
\begin{theorem}
    Let $(X^{(k)})_{k=1}^\infty$ be a sequence of random variables on $(\Omega,{\mathbb B}(\Omega))$ satisfying the following properties:
    \begin{enumerate}
	\item[(i)] $X^{(k)}(\Omega)$ is finite for all $k\in {\mathbb N}$.
	\item[(ii)] For $k,k'\in {\mathbb N}$ with $k<k'$ and $\omega_1,\omega_2\in \Omega$ it holds $X^{(k')}(\omega_1)<X^{(k')}(\omega_2)$ if $X^{(k)}(\omega_1)<X^{(k)}(\omega_2)$.
	\item[(iii)] For all $\omega_1,\omega_2\in\Omega$ there exists some $k\in {\mathbb N}$ with $X^{(k)}(\omega_1)\neq X^{(k)}(\omega_2)$.
    \end{enumerate}
    Then $h_\mu (T)=\lim\limits_{k\to\infty} h_\mu^{X^{(k)}}(T)$.
\end{theorem}

\section{Proof of Theorem 1}\label{proof}

The following discussion is aimed to characterize coincidence of the quantities $h$ and $h^\ast$ given in Section \ref{framework}.
We start with the generally valid inequality between $h$ and $h^\ast$.
\begin {lemma}\label{lemma3}
    It holds $h\leq h^\ast$.
\end {lemma}

\begin{proof}
It can be assumed that $h>0$. We fix some $\alpha >0$ with $h>\alpha$ and show that $h^\ast\geq\alpha$. Since $\alpha$ can be chosen arbitrarily near to $h$, this implies $h^\ast\geq h$.

Given $\beta >1$ with $h>\beta\,\alpha$, by \eqref{hd} and \eqref{h} there exists
some $d\in {\mathbb N}$ and some $n_d\in {\mathbb N}$ with $\frac{H_n(d)}{n}>\beta\,\alpha$ for all $n\geq n_d$.
Thus for all $n\geq \max \{n_d,\frac{d}{\beta-1}\}$ by \eqref{finer2} we obtain
\begin{align*}
    \frac{H_1(d+n-1)}{d+n-1}&\geq\frac{H_n(d)}{d+n-1}\\
			    &\geq\frac{H_n(d)}{(\beta-1)n+n-1}\\
			    &>\frac{H_n(d)}{\beta n}\\
			    &>\alpha,
\end{align*}
implying $h^\ast=\varlimsup\limits_{n \to \infty}\frac{H_1(d+n-1)}{d+n-1}\geq\alpha$.
\end{proof}

If $h=h^\ast$, then the upper limit in the definition of $h^\ast$ can be replaced by the usual limit as the following lemma shows.
\begin {lemma} If $h=h^\ast$, then
      \begin {equation}
	    h^\ast = \lim_{d \to \infty} \frac{H_1(d)}{d},
	    \label{eq:PE_without_sup}
      \end {equation}
      \label{th:PE_without_sup}
\end {lemma}
\begin{proof}
Given some $k\in \mathbb{N}$, by \eqref{finer2} it holds
\begin {equation*}
      \frac{H_1(k + n - 1)}{k + n - 1} \geq \frac{H_n(k)}{k + n - 1}
\end {equation*}
for all $n\in \mathbb{N}$, implying
\begin {equation*}
      \varliminf_{d \to \infty} \frac{H_1(d)}{d} = \varliminf_{n \to \infty} \frac{H_1(k+n-1)}{k+n-1} \geq \varliminf_{n \to \infty} \frac{H_n(k)}{k + n - 1} = h(k).
\end {equation*}
Therefore, we have
\begin {equation*}
%      \varliminf_{d \to \infty} \frac{H_1(d)}{d} \geq \varliminf_{k \to \infty} h(k)=h = h^\ast = \varlimsup_{d \to \infty} \frac{H_1(d)}{d},
      \varliminf_{d \to \infty} \frac{H_1(d)}{d} \geq \lim_{k \to \infty} h(k)=h = h^\ast = \varlimsup_{d \to \infty} \frac{H_1(d)}{d},
\end {equation*}
which shows \eqref{eq:PE_without_sup}.
\end {proof}

We come now to the proof of Theorem \ref{the01}.
We show equivalence of (i), (ii) and the following statement (iii') being equivalent to (iii).
\begin{itemize}
	\item [(iii')] For each $\varepsilon > 0$ there exists some $d_\varepsilon \in \mathbb{N}$ such that
		       for all $d \geq d_\varepsilon$ there is some $n_d \in \mathbb{N}$ with
		       \begin {equation*}
			    H_1(d + n - 1) - H_n(d) < (d+n - 1)\varepsilon \text{ for all }n \geq n_d.
		       \end {equation*}
\end{itemize}

Clearly, (iii) is stronger than (iii'). On the other hand, assume that (iii') is valid and $\varepsilon > 0$ is given. Fix some
$d_\varepsilon \in \mathbb{N}$ and for each $d \geq d_\varepsilon$ some $n_d \in \mathbb{N}$ with $n_d>d$ and
\begin {equation*}
    H_1(d + n - 1) - H_n(d) < (d+n - 1)\frac{\varepsilon}{2} \text{ for all }n \geq n_d.
\end {equation*}
Then, given $d\geq d_\varepsilon$, for all $n \geq n_d$ it holds
\begin {equation*}
    H_1(d + n - 1) - H_n(d) < ((n-1)+n - 1)\frac{\varepsilon}{2}= (n - 1)\varepsilon.
\end {equation*}

\noindent$(i) \Rightarrow (ii)$:
let $h=h^\ast$, let $\varepsilon>0$, and set $M_\varepsilon := h + \varepsilon/2$.
By Lemma \ref{th:PE_without_sup} and \eqref{h}, there exists some $d_\varepsilon \in \mathbb{N}$ such that
for all $d \geq d_\varepsilon$ it holds
\begin {equation}\label{heq1}
      \frac{H_1(d)}{d}<M_\varepsilon
\end {equation}
and
\begin {equation}\label{heq2}
      h - \frac{\varepsilon}{4} < h(d)
      \text{.}
\end {equation}

From \eqref{heq1} for all $d \geq d_\varepsilon$ and from \eqref{finer2} it follows
\begin {equation*}
      \frac{H_n(d)}{d + n - 1} \leq \frac{H_1(d + n - 1)}{d + n - 1} < M_\varepsilon
\end {equation*}
for all $d \geq d_\varepsilon$ and for all $n\in \mathbb{N}$, showing (\ref{eq:T1_statement1}).

Given $d\geq d_\varepsilon$, statement (\ref{eq:T1_statement1}) and inequality \eqref{heq2} imply existence of some $n_d\in\mathbb{N}$ with
$n_d \geq \frac{4 M_\varepsilon (d-1)}{\varepsilon}$ and
\begin {align*}
      h - \frac{\varepsilon}{4} &< \frac{H_n(d)}{n} = \frac{H_n(d)}{d + n - 1} + \frac{H_n(d)}{d + n - 1} \frac{d-1}{n} \\
				&< \frac{H_n(d)}{d + n - 1} + M_\varepsilon \frac{d-1}{n} 				\\      			
				&\leq \frac{H_n(d)}{d + n - 1} + \frac{\varepsilon}{4}
\end {align*}
for all $n\geq n_d$. Hence we have
\begin {equation*}
      M_\varepsilon - \varepsilon = h - \frac{\varepsilon}{2} < \frac{H_n(d)}{d + n - 1}
      \text{.}
\end {equation*}
for all $n\geq n_d$.
This shows \eqref{eq:T1_statement2}.

\noindent$(ii) \Rightarrow (iii')$:
let $\varepsilon >0$ and let $d_\varepsilon$ and $M_\varepsilon$ with \eqref{eq:T1_statement1} and \eqref{eq:T1_statement2} for all $d\geq d_\varepsilon$ be given.
Setting $n=1$ in \eqref{eq:T1_statement1}, one gets $\frac{H_1(d)}{d}< M_\varepsilon$ for all $d\geq d_\varepsilon$.
In particular, this provides for each $d\geq d_\varepsilon$
\begin{equation*}
\frac{H_1(d+n-1)}{d+n-1}< M_\varepsilon\text{ for all }n\in {\mathbb N}.
\end{equation*}
Combining this with \eqref{eq:T1_statement2} and \eqref{finer2} one obtains that for each $d\geq d_\varepsilon$ there exists some $n_d$ with
\begin {equation*}
      M_\varepsilon - \varepsilon< \frac{H_n(d)}{d + n - 1}\leq \frac{H_1(d+n-1)}{d+n-1}< M_\varepsilon \text{ for all }n \geq n_d.
\end {equation*}
Since $\varepsilon$ can be chosen arbitrarily small, this shows (iii').

\noindent$(iii') \Rightarrow (i)$:
assuming validity of (iii'), for each $\varepsilon > 0$ there exist some $d_\varepsilon \in \mathbb{N}$ such that
for all $d\geq d_\varepsilon$ there is some $n_d \in \mathbb{N}$ with
\begin {equation*}
    \frac{H_1(d + n - 1)}{d + n - 1} \leq \frac{H_n(d)}{d + n - 1}+\varepsilon
\end {equation*}
for all $n\geq n_d$.	
For $d\geq d_\varepsilon$ this implies
\begin {equation*}
    h^\ast=\varlimsup_{n \to \infty}\frac{H_1(d + n - 1)}{d + n - 1} \leq\varlimsup_{n \to \infty}\frac{H_n(d)}{d + n - 1}+\varepsilon =h(d)+\varepsilon,
\end {equation*}
hence by \eqref{h} we have $h^\ast\leq h+\varepsilon$. For $\varepsilon\to 0$, this provides $h^\ast\leq h$.
Now (i) follows by  Lemma \ref{lemma3}.

\section*{Acknowledgment}
This work was supported by the Graduate School for Computing in Medicine and Life Sciences
funded by Germany's Excellence Initiative [DFG GSC 235/1].

\end{document}